\documentclass[12pt]{article}

\input xy
\xyoption{all}

\usepackage[colorlinks=true,linkcolor=blue,citecolor=blue]{hyperref}

\usepackage{marvosym}
\usepackage{amssymb}
\usepackage{amsthm}
\usepackage{enumerate}
\usepackage[active]{srcltx}
\usepackage{amsmath}
\usepackage{nicefrac}
\usepackage{latexsym}
\usepackage{graphicx}
\usepackage{textcomp}
\numberwithin{equation}{section}
\theoremstyle{plain}
   \begingroup
         \newtheorem{theorem}[equation]{Theorem}
         
         \newtheorem{proposition}[equation]{Proposition}
         \newtheorem{cor}[equation]{Corollary}

   \endgroup

 \theoremstyle{definition}
   \begingroup
         \newtheorem*{notation}{Notation}

          \newtheorem*{appendice}{Appendix}
         
         \newtheorem{observation}[equation]{Observation}
         \newtheorem{formulas}[equation]{Formulas}
         
   \endgroup

\newcommand{\bb}{\mathbb}

\def\xto#1{\xrightarrow{#1}}

\title{\large A sensible estimate for the stability constant of the Lennard-Jones potential}

\author{\normalsize Sergio A. Yuhjtman - Universidad de Buenos Aires, Argentina}

\begin{document}
\maketitle

\begin{abstract} We show that the stability constant of the Lennard-Jones potential in $\bb R^3$,
$\Phi(x)=\|x\|_2^{-12}-2\|x\|_2^{-6}$, is smaller than $14.316$.
This is remarkably smaller than the best previously known bound. Our method is very elementary, and probably applicable 
to other similar potentials such as the Morse potentials. We also improve slightly, in the Lennard-Jones case, the lower bound for 
the minimum interparticle distance of an optimal $n$-particle configuration to $0.684$.
\end{abstract}

\section{Introduction} $ $

\vspace{-0.5cm}

Intermolecular forces are often well described through pair interactions given by a radial potential which is very repulsive at short distances,
weakly attractive at long distances and possess a unique equilibrium distance. One of the most extensively considered examples of such 
potentials is the \mbox{Lennard-Jones} potential, \cite{LJ}
$\Phi(x)=A_1\|x\|_2^{-12}-A_2\|x\|_2^{-6}$ (where $A_1,A_2 >0$). By rescaling domain and codomain we may assume, without loss of generality,
that $A_1=1$ and $A_2=2$:
$$\bb R^3 \xto \Phi \bb R \ \ \ \ \ \ \ \Phi(x)=\|x\|_2^{-12}-2\|x\|_2^{-6}$$
which attains the global minimum at $1$, and $\Phi(1)=-1$.
The term $\|x\|_2^{-6}$ can be theoretically justified as an interaction between dipoles, while
the $\|x\|_2^{-12}$ term resembles a hard-core interaction, since the potential grows fast as the norm of $x$ decreases 
from $1$.

For a finite configuration $Q \subset \bb R^3$, it is well-known that the energy per particle for the Lennard-Jones potential is bounded below. 
This condition is called ``stability''. An introduction to this concept can be found in \cite{Ruelle} section 3.2, jointly
with useful criteria to determine stability, e.g. proposition 3.2.8.
An important problem associated to these kind of potentials is to estimate the minimum possible energy per particle.
Thus, the ``stability constant'' is the minimum real $B \geq 0$ such that
$$\frac{1}{|Q|}\sum_{\substack{x,y \in Q \\ x \neq y}} \Phi(x-y) \geq -B$$
for every finite $Q \subset \bb R^3$. Here $|Q|$ denotes the cardinality of $Q$. In this paper
we prove $B \leq 14.316$ for the Lennard-Jones potential $\Phi$ (theorem \ref{main}).

One good reason that motivates the quest for a good estimate for $B$ is its direct implication on the convergence radius 
of the cluster expansion of the corresponding grandcanonical ensemble. This expansion shows that
the system behaves like a gas in the convergence region, and allows to compute exactly the thermodynamic
observables. In 1963, Penrose and Ruelle have independently shown, for certain class of stable potentials,
the convergence region $|\lambda| \leq (e^{2\beta B+1} C(\beta))^{-1}$, where $\lambda$ is the activity, $\beta$ the inverse temperature, 
and $C$ a function that depends on $\Phi$ but can be computed independently from $B$.
This classical result was recently improved by Morais, Procacci, Scoppola \cite{Aldo2} and de Lima, Procacci
\cite{Aldo1}, so the convergence region can now be written as:
$$|\lambda| \leq \frac{1}{e^{\beta B+1} \tilde C(\beta)}$$
where $\tilde C$ is a function similar to $C$. So far, the best theoretical upper bound for $B=B_{LJ}$ to their knowledge was $41.66$, provided by Schachinger, Addis, Bomze and Schoen \cite{Sch},
which is significantly higher than our bound $14.316$. Therefore, the present article 
enlarges drastically the proven convergence region for the Mayer series of the Lennard-Jones gas.

As already mentioned in \cite{Sch}, there are not many articles establishing rigorous results on the lowest energy
configurations for the Lennard-Jones or other similar potentials.
There has been, however, a lot of effort put at computationally 
finding the optimal configurations for manageable amounts of particles; see for example \cite{Locatelli}, \cite{110}. 
A reasonable lower bound for $B$ can be obtained by considering a particular configuration: a face centered cubic lattice.
This gives $B \geq 8.61$; see \cite{Sch} for the details. Another important feature is the minimal interparticle distance
for the optimal \mbox{$n$-particle} configuration. Here we improve the lower bound from $0.67985$ \cite{Sch}
to $0.684$. A brief description of the significance of this constant can be found in \cite{Vinko} section 1.1 and \cite{Sch}
section 1.

The proof of the main theorem, the upper bound for $B$, follows by composition of two main ideas.
The first of them is a very simple method to estimate
the energy of a given particle for a configuration with no particles closer than a given $a>0$. Consider disjoint spheres of radius
$\frac{a}{2}$ around each particle, and define a function whose average on every ball is larger than the energy 
of the corresponding particle. Then we can just replace the sum by the integral over the whole space divided by the volume of each ball.

The other idea is to estimate the global energy of a configuration with the method previously described, but with spheres of radius $\frac{1}{2}$
(actually we need a radius a little smaller). The crucial observation is that although these spheres can intersect each other, we can
control the contributions from the intersections and cancel each of them with the high energies generated by the two particles which
are so close together as to generate such an intersection.

In order to get good estimates, it is important to make a good choice for the function to be integrated. The average of a function
on a ball reminds one of the theory of harmonic and subharmonic functions. If we have a $C^2$ function with positive laplacian, then
it is subharmonic, and the value at a point is smaller than the average on every ball centered at it. The Lennard-Jones potential satisfies
that condition for $r \geq r_0 \approx 1.14$. For smaller values, we make a harmonic extension until a convenient radius. 
By the general nature of the arguments involved, it is reasonable to believe that the same strategy might be used to estimate 
the stability constant of other similar radial stable potentials such as the Morse potentials.

Let us now summarize the structure of the exposition. In section 2 we introduce some auxiliary functions jointly with basic properties 
and formulas that we will need later. In section 3 we find a good bound for the energy of a single particle depending
on the minimum distance separating any two particles of the configuration. As a corollary, we get our bound for the minimum
interparticle distance for a lowest energy $n$-particle configuration, valid for every $n \in \bb N$.
We also provide a proof for the existence of minimum energy configurations for a fixed number of particles, a previously known fact.
In section 4, we repeat the ideas from proposition \ref{mu} in order to get estimates involving larger balls, whose radii are almost $\frac{1}{2}$.
Finally, in section 5 we prove the main theorem, $B \leq 14.316$.

\section{Elementary facts}

As a convention, we will work with the minus-energy, instead of the energy:
$$\bb R_{>0} \xto h \bb R, \ \ \ h(r)=-r^{-12}+2r^{-6}$$
The following function will play an important role:
$$\tilde h= \chi_{(0,1]}+\chi_{(1,+\infty]}h$$
where $\chi_R$ is the characteristic function of the set $R$, defined as $1$ on $R$ and $0$ on its complement.

The function $h$ has a unique root at $\frac{1}{\sqrt[6]{2}} \simeq 0.8909$. It is negative for smaller values and positive
for the rest. Its derivative $h'$ has a unique root at $1$, it is positive before $1$ and negative thereafter.

\begin{formulas} \label{form}
 Consider in $\bb R^3$ two balls $B_1$ and $B_2$ with radii $r_1$ and $r_2$ respectively, whose centers are separated by 
 a distance $d \leq r_1+r_2$. The volume of the intersection is given by 
 $$a) \ \ \ |B_1 \cap B_2|=\frac{\pi}{12d}(r_1+r_2-d)^2(d^2+2d(r_1+r_2)- 3 (r_1-r_2)^2) $$
If we call $S_1$ and $S_2$ the surfaces of $B_1$ and $B_2$, the area of the portion of $S_1$ inside $B_2$ is:
$$b) \ \ \ |S_1 \cap B_2|= \pi \frac{r_1}{d}(r_1+r_2-d)(r_2-r_1+d)$$
  
  A quick way to prove these formulas is the following. Start showing that the height of the spherical cap
  $S_1 \cap B_2$ is equal to \mbox{$\frac{1}{2d}(r_1+r_2-d)(r_2-r_1+d)$}. Multiplying this by $2 \pi r_1$ we obtain the area.
  Then we can check the volume formula: the derivative of the volume with respect to $r_1$ must be $|S_1 \cap B_2|$
  and the volume at $r_1 = d-r_2$ must be $0$.
\end{formulas}

\begin{notation}
 As it was already done, for a subset $X \subset \bb R^3$, we denote by $|X|$ the volume, area or cardinality of $X$ 
 according to the dimension of $X$. For a point $x \in \bb R^3$, $\|x\|$ denotes the Euclidean norm.
\end{notation}

\begin{observation} $ $
 Calling $H(x)=h(\|x\|)$, its Laplacian satisfies:
 
 $\Delta H(x) \geq 0$ for $\|x\| \geq (\frac{11}{5})^{\frac{1}{6}}$.
 
 $\Delta H(x) = 0$ for $\|x\| = (\frac{11}{5})^{\frac{1}{6}}$.
 
 $\Delta H(x) \leq 0$ for $0 < \|x\| \leq (\frac{11}{5})^{\frac{1}{6}}$.
\end{observation}

\begin{proof}
 From the formula for the Laplacian in spherical coordinates, since there is no angular dependence, we have
 that the Laplacian of $H$ at a point $x$ with $\|x\|=r$ is:
 
$$\frac{1}{r^2} \partial_r(r^2 \partial_r h(r))= 12(-11r^{-14}+5r^{-8})$$
which is greater than $0$ if and only if $5r^6 > 11$ and lower than $0$ if and only if $5r^6<11$.
\end{proof}

\begin{proposition} \label{comparo} $ $

 (a) The function $\bb R_{>0} \xto t \bb R$, $t(r)=\frac{360}{121}(\frac{11}{5})^{\frac{1}{6}}r^{-1} - \frac{25}{11}$
 coincides with $h$ at $(\frac{11}{5})^{\frac{1}{6}}$. The first and second derivatives also coincide at that point.
 
 (b) 
$$t(r) > h(r) \mbox{ for } r<(\frac{11}{5})^{\frac{1}{6}}$$
$$t(r) < h(r) \mbox{ for } r>(\frac{11}{5})^{\frac{1}{6}}$$
\end{proposition}

\begin{proof}
 Part (a) is very easy to check. For part (b) we take the difference $t(r)-h(r)$ and turn it into a polynomial by multiplying by $r^{12}$.
 $$q(r)=r^{12}(t(r)-h(r))=- \frac{25}{11}r^{12} + \frac{360}{121}(\frac{11}{5})^{\frac{1}{6}}r^{11}- 2r^{6}+ 1 $$
 We know that $(\frac{11}{5})^{\frac{1}{6}}$ is a root of multiplicity at least $3$. By Descartes' rule of signs, the number
 of positive roots cannot exceed $3$, so there are no more positive roots. Now it is easy to check that $q$ is positive for
 $0<r<(\frac{11}{5})^{\frac{1}{6}}$ and negative for $r > (\frac{11}{5})^{\frac{1}{6}}$.
\end{proof}

\fontsize{9pt}{10pt}

\hspace{-.5cm}\includegraphics[scale=.4]{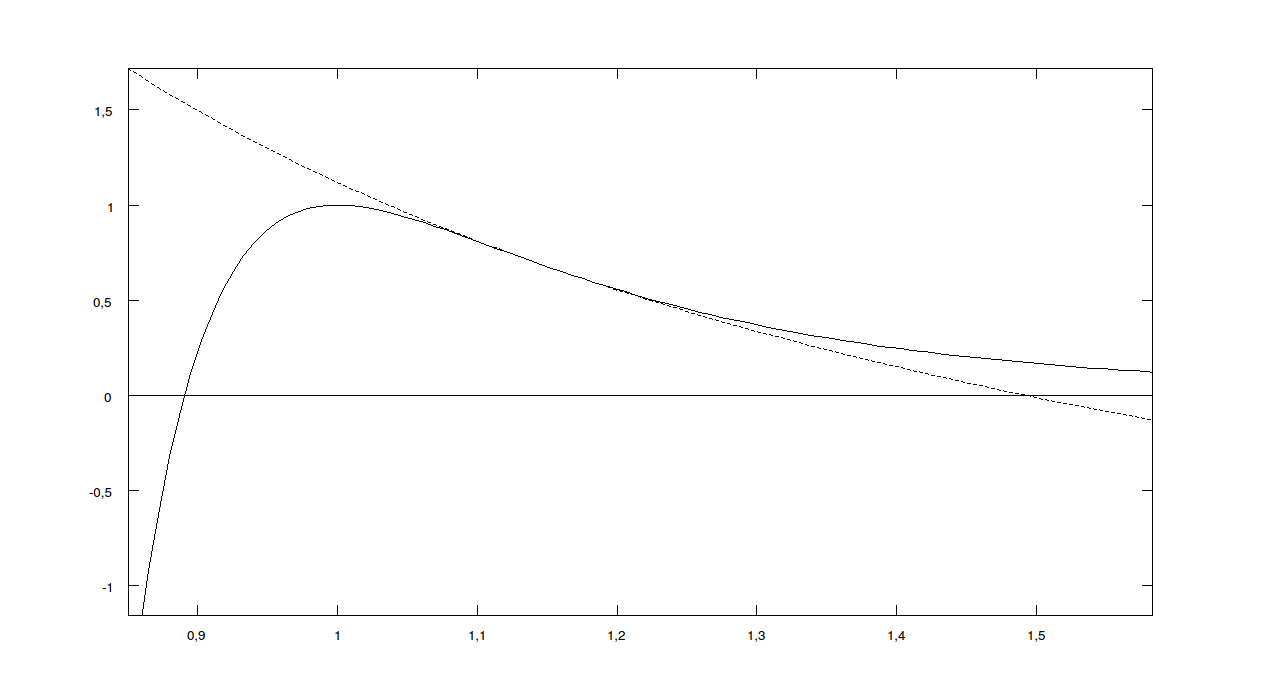}
\centerline{\hspace{0.1cm} \it Figure: The solid line represents the function $h$, the dashed line represents $t$.}

\normalsize

\begin{proposition} \label{fundamental}
The function $\bb R_{> 0} \xto \theta \bb R_{\geq 0}$ given by
\[\theta(r) = \left\{ \begin{aligned} & \frac{360}{121}(\frac{11}{5})^{\frac{1}{6}}r^{-1} - \frac{25}{11} & \ \ \ \ \ 0 < r \leq (\frac{11}{5})^{\frac{1}{6}}\\
& -r^{-12}+2r^{-6} & \ \ (\frac{11}{5})^{\frac{1}{6}} < r \end{aligned} \right.\]
satisfies
$$\tilde h(\|x\|) \leq \theta(x) \leq \frac{1}{|B_x|} \int_{B_x} \theta(\|y\|) dy$$
for every $x \in \bb R^3$ and $B_x$ any ball centered at $x$ such that $0 \notin B_x$.
\end{proposition}

\begin{proof}
 
 By the previous proposition, part (a), we have $\theta \in C^2(\bb R_{>0})$. 
 Consider the map $\bb R^3_{\neq 0} \xto{\Theta} \bb R_\geq 0$,
 $\Theta(x)=\theta(\|x\|)$. Reading $\Theta$ in spherical coordinates we conclude that $\Theta \in C^2(\bb R^3_{\neq 0})$.
 The laplacian $\Delta \Theta$ is $0$ for $\|x\| \leq (\frac{11}{5})^{\frac{1}{6}}$ and larger than $0$
 for $\|x\| > (\frac{11}{5})^{\frac{1}{6}}$, so $\Theta$ is subharmonic. 
 
 Besides, by part (b) of the previous proposition and $t(r) > 1$ for $r \leq 1$, we have $\tilde h \leq \theta$, so: 
 $$\tilde h(\|x\|) \leq \theta(\|x\|) \leq \frac{1}{|B_x|} \int_{B_x} \theta(\|y\|) dy$$
\end{proof}

\section{Estimate for the energy of a particle and minimum interparticle distance} $ $

\vspace{-.1cm}

It is useful to define, for every $a > 0$
$$\mu(a):= \sup \{\sum_{x \in Q} h(x) / Q \subset \bb R^3 \mbox{ finite}, \ d(x,y) \geq a \ \forall x,y \in Q\}$$
It is the supremum of the possible minus-energies from a point when considering configurations where every interparticle distance
is at least $a$.

In the following proposition, we show reasonable bounds for $\mu(a)$, where $a<0.7$. The second one is slightly sharper and its proof
demands more effort. This improved bound is not essential for our main theorem. However, we included it because it allows to 
increase the lower bound for the minimum interparticle distance in an optimal configuration, from approximately $0.679$ \cite{Sch} 
to $0.684$.

\begin{proposition} \label{mu} $ $

 \noindent I) For $0 \leq a \leq 0.7$ we have:
 $$\mu(a) \leq \frac{24}{a^3} \int_{0.54}^\infty \theta(w)w^2 dw < \frac{26.95}{a^3}$$
 II) For $0.6 \leq a \leq 0.7$
 
 $$\mu(a) \leq \frac{24}{a^3} \int_{0.64}^\infty \theta(w)w^2 dw < \frac{24.05}{a^3}$$

\end{proposition}

\begin{proof}
I)  Take any configuration $Q$ such that every distance is at least $a$. For every $x \in Q$ with norm larger than $0.89$,
 consider a ball $B_x$ of radius $\frac{a}{2}$ centered at $x$. These balls are disjoint, therefore, applying proposition \ref{fundamental}:
 
 $$\sum_{x \in Q} h(x) \leq \sum_{\substack{x \in Q \\ \|x\| \geq 0.89}} h(x) \leq 
 \sum_{\substack{x \in Q \\ \|x\| \geq 0.89}} \frac{1}{|B_x|} \int_{B_x} \theta(\|y\|) dy \leq$$
 $$\leq \frac{1}{\frac{4}{3}\pi.(\frac{a}{2})^3} \int_{\| y \| > 0.54} \theta(\|y \|)dy  =
 \frac{24}{a^3} \int_{0.54}^\infty \theta(r)r^2dr < \frac{26.95}{a^3}$$
 where the integral can be solved analytically. The number $0.54$ is equal to $0.89-\frac{0.7}{2}$.

II) For $r=\|x\| \geq 0.64+\frac{a}{2}$ we still can use $h(x) \leq \frac{1}{|B_x|} \int_{B_x} \theta(\|y\|) dy$ but to prove the assertion we need,
for $0.89 \leq \|x\| \leq 0.64+\frac{a}{2}$
 $$h(x) \leq \frac{1}{|B_x|} \int_{B_x} \theta^{0.64}(\|y\|) dy$$
where $\theta^{0.64}= \chi_{(0.64,+\infty)} \theta$.
 We can integrate through spherical coordinates. 
 Let $c=\frac{a}{2}$. For a radius $w$ between $0.64$ and $r+c$ we must consider the sphere centered at $0$ with radius $w$, 
and the area of its surface inside $B_x$. By the formula \ref{form} b), this area is equal to

$$\frac{\pi w}{r} (c + w - r)(c -w +r)= \frac{\pi w}{r} (-w^2 + 2r w + c^2 -r^2)  $$

Thus we have, calling $A=\frac{360}{121} (\frac{11}{5})^{\frac{1}{6}}$

$$\int_{B_x} \theta^{0.64}(\|y\|) dy \geq 
\frac{\pi}{r} \int_{0.64}^{r+c} (A w^{-1}-\frac{25}{11})w(-w^2 + 2r w + c^2 -r^2)dw \geq$$

$$ \geq \frac{\pi}{r} \int_{0.64}^{1.19} (A w^{-1}-\frac{25}{11})w(-w^2 + 2r w + c^2 -r^2)dw$$
We have used $h(w)>t(w)$ if $w > (\frac{11}{5})^{\frac{1}{6}}$ (proposition \ref{comparo}) for the first
inequality, and that $t(w)$ is positive if $w<1.49$ for the second. Notice that $1.19=0.89+\frac{0.6}{2}$

We compute the relevant primitives:
$$\int -w^2 + 2rw + (c^2-r^2) \ dw = -\frac{1}{3}w^3 + rw^2 + (c^2-r^2)w = \alpha(w)$$
$$\int -w^3 + 2rw^2 + (c^2-r^2)w \ dw = -\frac{1}{4}w^4 + \frac{2}{3}rw^3 + \frac{1}{2}(c^2-r^2)w^2 = \beta(w)$$

$$\int_{0.64}^{1.19} (A w^{-1}-\frac{25}{11})w(-w^2 + 2r w + c^2 -r^2)dw =$$
$$=A(\alpha(1.19)-\alpha(0.64))-\frac{25}{11}(\beta(1.19)-\beta(0.64)) \geq$$
$$\geq 0.7224 c^2-0.7225 r^2+1.2589r-0.5654$$
where the inequality is due to the truncation of decimal expressions.

Restoring the factors that multiply the integral, we must prove that:

$$\frac{3}{4c^3r} (0.7224 c^2-0.7225 r^2+1.2589r-0.5654) \geq h(r) \ \ \ \ \ (\clubsuit)$$

For $0.3 \leq c \leq 0.35$ and $0.89 \leq r \leq 0.64+c$. 

\underline{Claim:} the left hand side is decreasing in $c$ when $0.89 \leq r \leq 0.64+c$. 
To prove the claim, we differentiate the expression with respect to $c$ ignoring the factor
$\frac{3}{4 r}$. It gives:
$$-0.7224c^{-2}+(-3)(-0.7225 r^2+1.2589r-0.5654)c^{-4} \stackrel{?}{<} 0$$
$$-0.7224c^{2} \stackrel{?}{<} -2.1675 r^2+3.7767r-1.6962   $$
The maximum of the right hand side quadratic expression is attained at \mbox{$r=0.8712...<0.89$}. Thus,
the right hand side is larger than its value at $0.64+c$. After replacing $r$ by $0.64+c$,
we see that the remaining expression which is quadratic in $c$ has the correct sign.
$$0 < -1.4451 c^2+1.0023 c-0.16692   \ \ \ \ \mbox{ for } c \in [0.3,0.35]$$
and the claim has been proven.

Therefore, it suffices to check the inequality ($\clubsuit$) for $c=0.35$.
In this case we have that the left hand side is
$$\geq 22.021-12.639 r-\frac{8.343}{r}$$
The expression is decreasing for $r \geq 0.89$. Evaluating at $r=1$ we find that it is larger than $1$, so

$$22.021-12.639 r-\frac{8.343}{r} > 1 \geq h(r) \ \ \mbox{ for } 0.89 \leq r \leq 0.64+c < 1$$

Finally, $$\frac{24}{a^3} \int_{0.64}^\infty \theta(r)r^2dr < \frac{24.05}{a^3}$$
 
\end{proof}

The following statement seems to be well-known. We include it for the sake of completeness.

\begin{proposition}
 For every $n \in \bb N$ there is a configuration of $n$ points attaining the minimum possible energy
 for $n$ points.
\end{proposition}

\begin{proof}
 The function to maximize is the minus-energy, given by $\bb R^{3n}_{\neq} \xto F \bb R$, where 
 $$\bb R^{3n}_{\neq} = \{ Q=(x_1,...,x_n) / x_i \in \bb R^3, \ x_i \neq x_j\}, \ \ F(Q)=\sum_{i \neq j} h(\|x_i - x_j\|)$$
 The proposition follows by observing that $$\sup \{F(Q)\} = \sup \{F|_K(Q)\}$$ for a compact $K \subset \bb R^{3n}_{\neq}$.
 Indeed, we can take $$K=\{Q \in \bb R^{3n}_{\neq} / 0.65 \leq d(x_i,x_j) \leq 2(n-1), \ x_1=0\}$$
 To see this, we will show that any configuration has less or equal minus-energy than a configuration in $K$.
 
 Take a configuration $Q$. If there are $x_i, x_j$ such that $d(x_i,x_j) < 0.65$, then
 define $a$ as the minimum interparticle distance of $Q$. We have $a<0.65$. Take $x,y \in Q$ at distance $a$.
 By previous proposition, part I),
 $$\sum_{\substack{z \in Q \\ z \neq x}}h(\|x-z\|) \leq \mu(a)+h(a) \leq -a^{-12}+2a^{-6}+26.95a^{-3}$$
 We can easily check that $-a^{-12}+2a^{-6}+26.95a^{-3} < 0$ for $a<0.65$, since it is equivalent to 
 $2a^6+26.95a^9<1$, which follows evaluating at $0.65$ and by monotonicity. Therefore, the minus-energy of $x$ is negative, 
 so we can improve the configuration by taking the particle away 
 from the other particles. We can repeat the procedure until there are not two particles at distance less than $0.65$.
 
 Call $Q'$ the remaining configuration. If $d(x,z) >  2(n-1)$ for $x,z \in Q'$ consider the segment $xz$ 
 and the equispaced points $x=y_1,y_2,...,y_n=z$
 in that segment. The distance between two consecutive points is greater than 2.
 Take planes $P_1,...,P_n$ orthogonal to $xz$, $P_k$ passing through $y_k$. Consider $R_k \subset \bb R^3$ ($1 \leq k \leq n-1$)
 the open regions between the planes $P_k$ and $P_{k+1}$.
 Clearly, there is at least one of those regions, say $R_k$, without points of $Q'$. 
 We can reduce by $1$ the distance between $P_k$ and $P_{k+1}$ (translating accordingly the points of $Q'$) incrementing
 the minus-energy provided by every pair of particles. Notice that through this transformation, the distance between 
 $x$ and $z$ decreases by $1$, and no interparticle distance increases. In addition, it does not produce a pair separated 
 by less than $0.65$ if it was not there before.
 After repeating the procedure as many times as possible, we necessarily reach a configuration $Q''$ satisfying 
 $0.65 \leq d(x,y) \leq 2(n-1)$ $\forall x,y \in Q''$.
 
 Finally, it is clear that we can translate so that $x_1=0$.

\end{proof}

\begin{cor} \label{interparticle}
 An optimal (lowest energy) configuration $Q \subset \bb R^3$ of $n$ points satisfies
$$d(x,y) > 0.684  \ \ \forall x,y \in Q$$ 
\end{cor}

\begin{proof}
 Take an optimal configuration $Q$, which exists by the previous proposition, and call the minimum distance between two 
 particles $a$. Take $x,y \in Q$ with $\|x-y\|=a$.
 As seen in the previous proof, if $a<0.65$, the particle $x$ necessarily has positive energy, so
 the configuration cannot be optimal. Therefore, $a \geq 0.65$ and we can use proposition \ref{mu} II).
  $$0 \leq \sum_{\substack{z \in Q \\ z \neq x}} h(\|z-x\|) \leq \mu(a)+h(a) \leq -a^{-12}+2a^{-6}+24.05a^{-3}$$
   $$24.05a^9+2a^6 \geq 1$$
 By monotonicity, we see that this does not hold for $a \leq 0.684$.
\end{proof}

\section{Balls with almost 1/2 radius} $ $

\vspace{-.3cm}
The proof of the main theorem needs to consider balls with a radius $c$ as close to $\frac{1}{2}$ as possible, in order to avoid
the factor $\frac{1}{c^3}$ for a small $c$. This is possible for $c=0.49$, as will be shown. Then the challenge is to find
a positive function whose average value on every ball with such a radius is larger than the value of $\tilde h$ at its center (it suffices
to consider $\|x\|>0.684$, where $x$ is the center of the ball) and at the same time its integral on $\bb R^3$ should be as low as possible.

\begin{proposition} \label{unmedio}
The function $\bb R_{> 0} \xto {\theta^{0.54}} \bb R_{\geq 0}$ given by
$\theta^{0.54}=\chi_{(0.54,\infty)} \theta$ satisfies
$$\tilde h(\|x\|) \leq \frac{1}{|B_x|} \int_{B_x} \theta^{0.54}(\|y\|) dy$$
for every $x \in \bb R^3$ with $\|x\| \geq 0.51$ and $B_x$ the ball centered at $x$ with radius $c=0.49$.
\end{proposition}

\begin{proof}

We divide into three regions according to $\|x\|=r \in \bb R_{\geq_0}$.

\underline{Region 1:} $r \geq 1.03$. It holds by subharmonicity of $\Theta(x)=\theta(\|x\|)$ and $\tilde h \leq \theta$, i.e. proposition \ref{fundamental}.

\underline{Region 2:} $0.51 \leq r \leq 0.9$. 

Again, as in the proof of proposition \ref{mu}, we can treat the integral of $\theta^{0.54}$ with spherical
coordinates, taking advantage of the absense of angular dependence of $\Theta$. 
We use the formula \ref{form} b) for the surface of a sphere inside
another sphere. The ball $B_x$ reaches the norms $0.54$ and $1$, therefore:
$$\int_{B_x} \theta^{0.54}(\|y\|) dy \geq \frac{\pi}{r} \int_{0.54}^{1} (A w^{-1}-\frac{25}{11})w(-w^2 + 2r w + c^2 -r^2)dw$$
After a straightforward calculation as in proposition \ref{mu} we reach:
$$\int_{0.54}^{1} (A w^{-1}-\frac{25}{11})w(-w^2 + 2r w + c^2 -r^2)dw  \geq -0.7558 r^2+1.127 r-0.2516$$
(the inequality is only due to truncation). Now it only remains to check
$$\frac{3}{4(0.49)^3r}(-0.7558 r^2+1.127 r-0.2516)>1$$
for $0.51 \leq r \leq 0.9$, that can be reduced to an inequality for a quadratic expression.

\bigskip

\underline{Region 3:} $0.9 \leq r \leq 1.03$

We proceed as before, but limit the integral between $0.54$ and $1.39$. We make use of $h(r) \geq t(r)>0$ for 
$(\frac{11}{5})^{\frac{1}{6}} \leq r \leq 1.49$ (proposition \ref{comparo} (b)).
$$\int_{B_x} \theta^{0.54}(\|y\|) dy \geq \int_{0.54}^{1.39} (A w^{-1}-\frac{25}{11})w(-w^2 + 2r w + c^2 -r^2)dw \geq$$
$$\geq -1.0199 r^2+1.7357 r-0.5418$$
and it suffices to verify
$$\frac{3}{4(0.49)^3r}(-1.0199 r^2+1.7357 r-0.5418)>1$$
for $0.9 \leq r \leq 1.03$.

\end{proof}

\section{Lower bound for the average energy} $ $

\begin{theorem} \label{main}
  Let $\bb R^3 \xto \Phi \bb R, \ \mbox{$\Phi(x)=\|x\|_2^{-12}-2\|x\|_2^{-6}$}$ be the Lennard-Jones potential.
   Every finite configuration $Q \subset \bb R^3$ satisfies 
   $$\frac{1}{|Q|} \sum_{\substack{x,y \in Q \\ x \neq y}} \Phi(\|x-y\|) > -14.316$$
   so the stability constant $B$ is at most $14.316$.
\end{theorem}

\begin{proof}
Define 
$$H_x(y)=h(\|y-x\|), \ \tilde H_x(y)=\tilde h(\|y-x\|), \ \Theta^{0.54}_x(y)=\theta^{0.54}(\|y-x\|)$$ for $x,y \in \bb R^3$.
We also consider $H=H_0$, $\tilde H=\tilde H_0$, $\Theta^{0.54}=\Theta^{0.54}_0$. Recall $\theta^{0.54}=\chi_{(0.54,+\infty)}\theta$.

We can assume that $Q$ is an optimal configuration of $|Q|$ points. Therefore, by corollary \ref{interparticle} we have $\|x-y\| > 0.684$ for $x,y \in Q$.
For each $x \in Q$, consider $B_x$ the ball centered at $x$ with radius $0.49$.

Let us say $y \sim z$ whenever $\|y-z\| < 0.98$ and $y \neq z$. In this case we have a nontrivial intersection $N_{yz}:=B_y \cap B_z$.

Take $x_0 \in Q$.

$$\sum_{\substack{y \in Q \\ y \neq x_0}} H_{x_0}(y) \leq 
\sum_{\substack{y \in Q \\ y \neq x_0}} \tilde H_{x_0}(y) + \sum_{\substack{y \in Q \\ y \sim x_0}} (H_{x_0}(y)-1)$$
The inequality is due to the points at a distance between $0.98$ and $1$ from $x_0$. Applying proposition \ref{unmedio},
$$\sum_{\substack{y \in Q \\ y \neq x_0}} \tilde H_{x_0}(y) \leq \sum_{\substack{y \in Q \\ y \neq x_0}} \frac{1}{|B|} \int_{B_y} \Theta^{0.54}_{x_0}
\leq \frac{1}{|B|} \int_{\bb R^3} \Theta^{0.54} + \frac{1}{|B|} \sum_{\substack{y \sim z \\ x_0 \neq y,z \in Q}} \int_{N_{yz}} \Theta^{0.54}_{x_0}$$

The last inequality may be intuitively clear, but a simple formal proof is possible. For $0 \leq k \leq |Q|-1$,
call $E_k \subset \bb R^3$ the set of points belonging to exactly $k$ balls $B_y$ with $x_0 \neq y \in Q$.
At the left, the integral on $E_k$ is counted with multiplicity precisely $k$, while at the right it is counted
with multiplicity $1+\binom{k}{2}$, which is greater or equal than $k$ for every integer $k \geq 0$.

Now we perform the sum $\sum_{x_0 \in Q}$

$$2 \sum_{\substack{x,y \in Q \\ x \neq y}} h(\|x-y\|)= \sum_{x_0 \in Q} \sum_{\substack{y \in Q \\ y \neq x_0}} H_{x_0}(y) \leq $$

$$ \leq \frac{|Q|}{|B|} \int_{\bb R^3} \Theta^{0.54} + \sum_{\substack{y \sim z}} 
\Big( \frac{1}{|B|} \sum_{\substack{x \in Q \\ x \neq y,z}} \int_{N_{yz}} \Theta^{0.54}_{x} + 2(H(y-z) -1) \Big)$$
We will show that the parenthesis in the second term is less or equal than $0$ for every pair $y \sim z$.
First, we can control $\frac{1}{|B|} \sum_{x \neq y,z} \int_{N_{yz}} \Theta^{0.54}_x$ by applying proposition \ref{fundamental}
to the balls of radius $0.684/2$ centered at every $x \in Q$. We have, for every $w \in \bb R^3$,
$$\sum_{x \in Q} \Theta^{0.54}_{x}(w)=\sum_{x \in Q} \Theta^{0.54}_w(x) \leq \frac{24 \int_0^\infty \theta(r)r^2dr}{0.684^3} < \frac{36}{0.684^3}<113$$
Interchanging the sum with the integral, we find
$$\frac{1}{|B|} \sum_{x \neq y,z} \int_{N_{yz}} \Theta^{0.54}_{x} \leq \frac{1}{|B|} 113 |N_{yz}| = $$
$$113 \frac{1}{2 \times 0.98^3} (0.98-d_{yz})^2 (d_{yz}+2 \times 0.98)$$
where $d_{yz}=\|y-z\|$. We have used formula \ref{form}a) for the volume of $N_{yz}$.

With the help of a computer it can be easily visualized that
$$\frac{113}{2} \frac{1}{(0.98)^3} (0.98-d)^2 (d+1.96) + 2(h(d)-1) < 0$$
for every $d$ such that $0 < d \leq 0.98$. See the appendix for the proof of this inequality.

Now we can conclude

$$\frac{1}{|Q|} \sum_{\substack{x,y \in Q \\ x \neq y}} h(\|x-y\|) \leq \frac{1}{2|B|} \int_{\bb R^3} \Theta^{0.54} =$$
$$=\frac{12}{(0.98)^3} \times \int_{0.54}^\infty \theta(r) r^2 dr < 14.316$$

\end{proof}

 \begin{appendice}
 We must prove the inequality
 $$\frac{113}{2} \frac{1}{(0.98)^3} (0.98-d)^2 (d+1.96) + 2(h(d)-1) < 0 \ \ \ (0 \leq d \leq 0.98)$$ 
 Define $P(d)$ as the polynomial which results from multiplying the left hand side by $d^{12}$.
$$P(d)=c_1 d^{15}+c_2 d^{13}+111 d^{12}+4d^6-2$$
$$c_1=\frac{113}{2(0.98)^3} \ \ \ c_2=-\frac{8475}{49}$$
One way would be to apply Sturm's method to this polynomial to show that it has no roots between $0$ and $0.98$.
However, we can simplify. It suffices to show that $P'>0$ between $0$ and $0.98$
and evaluating at $0.98$. So the problem reduces to
$$R(d)=15c_1d^9+13c_2d^7+1332d^6+24 > 0$$
Again, we could apply Sturm's algorithm to this lower degree polynomial. But we can also efficiently find the minimum of $Q$
in our region, since
$$R'(d)=d^5(d-\rho_1)(d-\rho_2)(d-\rho_3)$$
where $\rho_1 \simeq -1.59958$, $\rho_2 \simeq 0.647647$, $\rho_3 \simeq 0.951934$. The minimum of $Q$ between $0$ and $0.98$
is attained at $\rho_3$, where it is positive.
\end{appendice}

{\bf Acknowledgement.} I am totally grateful to Aldo Procacci. He taught me about this problem and much more.

\bigskip

This work has been supported by the argentinian state organism Consejo Nacional de Investigaciones Cient\'ificas y T\'ecnicas (CONICET).

\bibliographystyle{amsplain}
\bibliography{statmech}{}

{\bf Sergio Andr\'es Yuhjtman}

sergioyuhjtman@gmail.com

Universidad de Buenos Aires, FCEN, Departamento de Matem\'atica

Intendente Guiraldes 2160, Ciudad Universitaria, Pabell\'on I, C1428EGA

Ciudad Aut\'onoma de Buenos Aires, Argentina

\end{document}